\documentclass[10pt, amssymb, aps, pra]{revtex4-2}
\usepackage{braket}
\usepackage{graphicx}
\usepackage{amsthm}
\usepackage{amsmath}
\usepackage{amssymb}
\usepackage{color}
\usepackage{soul}
\usepackage{url}
\usepackage{mathptmx}
\usepackage{times}
\usepackage{txfonts}
\usepackage[margin=1in]{geometry}
\usepackage{tikz}
\usepackage{tikz-cd}
\usetikzlibrary{arrows.meta}
\usepackage{hyperref}

\hypersetup{
    colorlinks=true,
    allcolors=black
}

\tikzset{every picture/.style={line width=0.75pt}} 

\DeclareMathOperator{\spn}{span}

\begin{document}

\title{Two-site Bose-Hubbard hopping and Schr\"{o}dinger cat states}
\author{Madeline Berezowski$^{1}$, Artur Sowa$^{1}$, Jonas Fransson$^{2}$}
\affiliation{$^{1}$Department of Mathematics and Statistics, University of Saskatchewan, Canada}
\affiliation{$^{2}${Department of Physics and Astronomy, University of Uppsala, Sweden}}

\maketitle
\newtheorem{definition}{Definition}
\newtheorem{theorem}{Theorem}[section]
\newtheorem{proposition}{Proposition}
\newtheorem{lemma}{Lemma}
\newtheorem{corollary}{Corollary}
\newtheorem{algorithm}{Algorithm}
\newtheorem{conjecture}{Conjecture}

\begin{center}

 Abstract

 \end{center}

The Bose-Hubbard Hamiltonian can be simplified to have only two lattice sites, in which case the system being described is referred to as a dimer. Due to its structure, the hopping term of the dimer Hamiltonian enjoys invariance in a family of subspaces indexed by a whole number $k$, each subspace corresponding to a system of only $k$ particles. We have invented an inductive argument using the bosonic canonical commutation relations to find the eigenvalues and eigenvectors of the dimer hopping Hamiltonian in its $k$-particle subspaces. In particular, this Hamiltonian, when restricted to one of the $k$-particle subspaces, is exactly the spin projection operator along the $x$-axis, where the number of particles $k$ in the dimer system yields the projection matrix for spin quantum number $s=k/2$. Thus, a new method for computing the eigenvalues and eigenvectors of the $x$-axis spin projector has been unearthed. We use the explicit construction to study the dynamics of coherent states induced by the square of the dimer hopping hamiltonian. We find that it generates Schr\"{o}dinger cat states in the two-site setting. 

\vspace{.2cm}

    \vspace{.2cm}

\section{Introduction}

The Bose-Hubbard model describes the physics of spinless interacting bosons in a lattice. There is a multitude of papers that investigate this system. The case where there are only two lattice sites, known as a double-well system or a dimer, is also well-studied, e.g. \cite{dimer01}, \cite{Zhang_Rieger}, \cite{dimer02}. Such systems are also studied numerically, e.g., \cite{exact_diag_bh}. 

The number theoretic methods for the study of the Bose-Hubbard model, introduced in \cite{number_theoretic_bose_hubbard}, \cite{nonlocal_coherent_states}, are necessary to study systems with infinitely many degrees of freedom. While it would not be necessary in the discussion of the dimer case, we adhere to this framework to be consistent with the general approach. The main result is an explicit characterization of the eigenvectors of the hopping term, given in Theorem \ref{TheorMain}. In particular, it reproduces the well known eigenvalues of the spin matrices. The construction enables us to express the coherent states in the form of superpositions of the hopping term's eigenvectors. It leads to the conclusion that the dynamics generated by the square of the hopping term periodically evolves coherent states into cat states. A study of such phenomena was initiated in the classical articles, \cite{Milburn} and \cite{Yurke-Stoler}. Our approach, developed in Section \ref{Section_Cat}, follows the framework introduced in the latter reference. The Schr\"{o}dinger cat states are foundational for bosonic quantum computing. In fact, the two site case is of special importance, see e.g. \cite{Chen}.   

We consider the Hamiltonian operator, which describes the energy of the system, and whose domain is the separable Hilbert space called the bosonic Fock space $\mathbb{H}^\odot$. This space allows us to describe a system with any number of bosons. If the system has only one particle, and that particle is located in the $j$th lattice site, the corresponding state is described by $\ket{p_j}$, where $p_j$ is the $j$th prime number. Configurations with only one particle are described by the subspace $\mathbb{H}_{\text{SP}}=\spn\{\ket{p}: p \text{ prime}\}$. Configurations with $k$ particles are described using the symmetric tensor product by $\mathbb{H}_{\text{SP}}^{\odot k}$. This space is spanned by vectors $\ket{n}$ for which $n$ has $k$ prime factors. The multiplicity of prime $p_j$ in the prime decomposition of $n$ is equal to the number of bosons at the $j$th lattice site. Putting this all together, we have that
\begin{equation}\label{Fock}
\mathbb{H}^\odot = \bigoplus_{k=0}^\infty \mathbb{H}_{\text{SP}}^{\odot k},
\end{equation} 
where $\mathbb{H}_{\text{SP}}^{\odot 0}=\mathbb{C}$. Note that $\mathbb{H}^\odot$ is isomorphic to $\ell_2(\mathbb{N})$.

The lowering operator $\hat{a}_p$ acts on the state $\ket{n}$ via
\begin{equation}
\hat{a}_p\ket{n} = \sqrt{a_p(n)}\ket{\frac{n}{p}}
\end{equation}
where $a_p(n)$ is the multiplicity of $p$ in the prime decomposition of $n$. This operator represents removing a boson from site $p$. Its adjoint, the creation operator acts as follows
\begin{equation}
\hat{a}_p^\dagger\ket{n} = \sqrt{a_p(n)+1}\ket{np},
\end{equation}
representing the addition of a boson to site $p$.

These operators satisfy the Bosonic Canonical Commutation Relations
\begin{equation}\label{BCCM}
[\hat{a}_p, \hat{a}_q^\dagger] = \delta_{p,q}, \quad [\hat{a}_p, \hat{a}_q]=0
\end{equation}
for all primes $p$ and $q$.

Recall that the number operators are defined via $\hat{N}_p = \hat{a}_p^\dagger \hat{a}_p$, and the total number operator is given by $\hat{N} = \sum_p \hat{N}_p$. The Bose-Hubbard Hamiltonian is
\begin{equation}
\mathcal{H} = \sum_{n=1}^\infty\Bigg(\frac{U}{2}\hat{N}_{p_n}(\hat{N}_{p_n}-1) - \mu \hat{N}_{p_n} - t(\hat{a}_{p_n}^\dagger\hat{a}_{p_{n+1}} + \hat{a}_{p_{n+1}}^\dagger\hat{a}_{p_n})\Bigg).
\end{equation}
The first term represents the on-site interaction, whose strength is controlled by $U$, and the second term describes the chemical potential $\mu$. The third term dictates hopping, or quantum tunnelling, between adjacent sites. The probability that a boson tunnels from one site to a neighbouring site is related to the hopping amplitude $t$.

\section{The Bose-Hubbard Dimer}
Consider the number theoretic Bose-Hubbard model with only two lattice sites, which are labelled by the prime numbers 2 and 3. Dropping the hopping amplitude factor, the hopping term of the Hamiltonian (which describes the energy asssociated  with quantum tunnelling) is given by
 \begin{equation}\label{hopping hamiltonian}
    \mathcal{H} = \hat{a}_2^\dagger\hat{a}_3 + \hat{a}_3^\dagger\hat{a}_2. 
 \end{equation}
In the case where the hopping term dominates over the other terms in the Hamiltonian, the system is in the superfluid phase. We restrict our attention to the superfluid phase with two lattice sites.

Each $k$-particle subspace is invariant under the action of this Hamiltonian, i.e. 
\begin{equation}\label{k particle subspace}
    \mathbb{H}^{\odot k} = \spn\{\ket{2^\alpha3^{k-\alpha}}: \alpha = 0,1,2,\dots,k\}.
\end{equation}
This is easy to see because
\begin{equation}
\hat{a}_2^\dagger\hat{a}_3\ket{2^\alpha3^{k-\alpha}} = \sqrt{(\alpha+1)(k-\alpha)} \ket{2^{\alpha+1}3^{k-\alpha-1}}
\end{equation}
and
\begin{equation}
 \hat{a}_3^\dagger\hat{a}_2\ket{2^\alpha3^{k-\alpha}}= \sqrt{\alpha(k-\alpha+1)}\ket{2^{\alpha-1}3^{k-\alpha+1}}. 
\end{equation}
In the $k$-particle subspace, the Hamiltonian is represented in the distinguished basis by the $(k+1)\times(k+1)$ tridiagonal matrix
\begin{equation}\label{matrix form}
    \mathcal{H}\big\lvert_{\mathbb{H}^{\odot k}} = \begin{pmatrix} 
    0 & \sqrt{k} \\
    \sqrt{k} & 0 & \sqrt{2(k-1)}\\
    & \sqrt{2(k-1)} & 0 & \sqrt{3(k-2)}\\
    & & \sqrt{3(k-2)} & 0 &\\
    & & & & \ddots\\
    & & & & & \sqrt{k}\\
    & & & & \sqrt{k} & 0 
    \end{pmatrix}
\end{equation}
for $k=0,1,2,3, \dots$\space. Equivalently, we can write
\begin{equation}
    \mathcal{H}\big\lvert_{\mathbb{H}^{\odot k}} = \sum_{\alpha=0}^k\Big(\sqrt{(k-\alpha)(\alpha+1)}\ket{2^{\alpha+1}3^{k-\alpha-1}}+ \sqrt{\alpha(k-\alpha+1)}\ket{2^{\alpha-1}3^{k-\alpha+1}}\Big)\bra{2^\alpha3^{k-\alpha}},
\end{equation}
This is the spin projection operator in the $x$-direction $S_x$, where we have spin quantum number $s=k/2$.

\section{Eigenvalues and eigenvectors}
The explicit formula for the eigenvalues and special matrices are typically derived using recurrence relations satisfying their characteristic polynomial. In the case of $S_x$, the eigenvalues are determined using its unitary equivalence with $S_z$. While the eigenvectors of $S_x$ are well-known, we provide a proof with an entirely new structure. This proof will enable firther investigations in the following sections. 

First, define
\begin{equation}\label{def_c_and_d_clean}
  \hat{c} =  \frac{\hat{a}_2+\hat{a}_3}{\sqrt{2}}\quad\mbox{ and }\quad
  \hat{d} =  \frac{\hat{a}_2-\hat{a}_3}{\sqrt{2}}.
\end{equation}
Using (\ref{BCCM}), we obtain
\begin{equation}\label{com_cd}
  [\hat{c},\hat{d}]= [\hat{c},\hat{d}^\dagger]= 0, \quad [\hat{c}, \hat{c}^\dagger] = [\hat{d}, \hat{d}^\dagger] = 1.
\end{equation}
Furthermore, it follows that
\begin{equation}\label{com_Hd}
    [\mathcal{H}, \hat{d}^\dagger] =-\hat{d}^\dagger, \quad   [\mathcal{H}, \hat{c}^\dagger] =\hat{c}^\dagger
\end{equation}
The following theorem characterizes the eigenvalues and eigenvectors. Note, however, that the eigenvectors given by (\ref{eigenvectors}) are \emph{not normalized} to have norm $1$.

\begin{theorem}\label{TheorMain}
The eigenvectors of $\mathcal{H}\lvert_{\mathbb{H}^{\odot k}}$ are 
\begin{equation}\label{eigenvectors}
(\hat{c}^\dagger)^m(\hat{d}^\dagger)^{k-m}\ket{2^03^0}, \quad\quad\quad m=0,1,2,\dots,k,
\end{equation}
with corresponding eigenvalues $-k+2m$.
\end{theorem}

\begin{proof}
We consider the action of $\mathcal{H}$ itself, as in Equation (\ref{hopping hamiltonian}), rather than any of its restrictions to invariant subspaces. An inductive proof will be employed.

The eigenvalue equation for the base case $k=0$ is trivially true,
\begin{equation}
    \mathcal{H}\ket{2^03^0} = \hat{a}_2^\dagger\hat{a}_3\ket{2^03^0} + \hat{a}_3^\dagger\hat{a}_2\ket{2^03^0} =0 = 0\ket{2^03^0}. 
 \end{equation}
For the inductive step, assume for arbitrary $k$,
\begin{equation}
    \mathcal{H}(\hat{c}^\dagger)^m(\hat{d}^\dagger)^{k-m}\ket{2^03^0} = (-k + 2m)(\hat{c}^\dagger)^m(\hat{d}^\dagger)^{k-m}\ket{2^03^0}.
\end{equation}
To compute $\mathcal{H}(\hat{c}^\dagger)^m(\hat{d}^\dagger)^{k+1-m}\ket{2^03^0}$, first note that $[\hat{c}^\dagger, \hat{d}^\dagger] = 0$ because $[\hat{a}_p^\dagger,\hat{a}_q^\dagger]=0$ for any prime numbers $p, q$. Therefore, 
\begin{equation}
    \mathcal{H}(\hat{c}^\dagger)^m(\hat{d}^\dagger)^{k+1-m}\ket{2^03^0} = \mathcal{H}\hat{d}^\dagger(\hat{c}^\dagger)^m(\hat{d}^\dagger)^{k-m}\ket{2^03^0}.
\end{equation}
Furthermore, (\ref{com_Hd}) implies $\mathcal{H}\hat{d}^\dagger = \hat{d}^\dagger(\mathcal{H}-1)$. This yields
\begin{equation}
\begin{split}
    \mathcal{H}(\hat{c}^\dagger)^m(\hat{d}^\dagger)^{k+1-m}\ket{2^03^0} &= \hat{d}^\dagger(\mathcal{H}-1)(\hat{c}^\dagger)^m(\hat{d}^\dagger)^{k-m}\ket{2^03^0}\\
    &=\hat{d}^\dagger(-k+2m)(\hat{c}^\dagger)^m(\hat{d}^\dagger)^{k-m}\ket{2^03^0} - (\hat{c}^\dagger)^m(\hat{d}^\dagger)^{k+1-m}\ket{2^03^0}\\
    &= \big(-(k+1) + 2m\big)(\hat{c}^\dagger)^m(\hat{d}^\dagger)^{k+1-m}\ket{2^03^0}.
\end{split}
\end{equation}
It is also necessary to consider the case $m=k+1$, as seen in Figure \ref{dimer_eigen_table}. That is
\begin{equation}
    \mathcal{H}(\hat{c}^\dagger)^{k+1}(\hat{d}^\dagger)^{k+1-(k+1)}\ket{2^03^0} = \mathcal{H}(\hat{c}^\dagger)^{k+1}\ket{2^03^0}
\end{equation}
Also, using (\ref{com_Hd}) again, we obtain $\mathcal{H}\hat{c}^\dagger = \hat{c}^\dagger(\mathcal{H}+1)$.
Thus, by the inductive assumption,
\begin{equation}
\begin{split}
    \mathcal{H}(\hat{c}^\dagger)^{k+1}\ket{2^03^0} &= \mathcal{H}\hat{c}^\dagger(\hat{c}^\dagger)^k\ket{2^03^0}\\
    &=\hat{c}^\dagger\mathcal{H}(\hat{c}^\dagger)^k\ket{2^03^0} + (\hat{c}^\dagger)^{k+1}\ket{2^03^0}\\
    &=k(\hat{c}^\dagger)^{k+1}\ket{2^03^0} + (\hat{c}^\dagger)^{k+1}\ket{2^03^0}\\
    &= (k+1)(\hat{c}^\dagger)^{k+1}\ket{2^03^0}.
\end{split}
\end{equation}
This completes the proof. 
\end{proof}
\vspace{.5cm}

\noindent
\emph{Remark 1.} Note that using (\ref{def_c_and_d_clean}), the Hamiltonian may be represented in the form:
\begin{equation}\label{theH_in_cd}
  \mathcal{H} =  \hat{c}^\dagger \hat{c} - \hat{d}^\dagger \hat{d}.
\end{equation}
\vspace{.5cm}

\noindent
\emph{Remark 2.}
The essence of the above proof and the iterative relationship between eigensystems is captured by Fig. \ref{cascade}.

\begin{figure}[ht]\label{cascade}
    \centering
    \begin{tikzpicture}[x=0.75pt,y=0.75pt,yscale=-1,xscale=1]

\draw    (180,90) -- (179.88,116.07) ;
\draw [shift={(179.87,118.07)}, rotate = 270.27] [color={rgb, 255:red, 0; green, 0; blue, 0 }  ][line width=0.75]    (10.93,-3.29) .. controls (6.95,-1.4) and (3.31,-0.3) .. (0,0) .. controls (3.31,0.3) and (6.95,1.4) .. (10.93,3.29)   ;
\draw    (180,160) -- (179.9,181.09) -- (179.88,186.07) ;
\draw [shift={(179.87,188.07)}, rotate = 270.27] [color={rgb, 255:red, 0; green, 0; blue, 0 }  ][line width=0.75]    (10.93,-3.29) .. controls (6.95,-1.4) and (3.31,-0.3) .. (0,0) .. controls (3.31,0.3) and (6.95,1.4) .. (10.93,3.29)   ;
\draw    (240,90) -- (280.18,111.07) ;
\draw [shift={(281.95,112)}, rotate = 207.67] [color={rgb, 255:red, 0; green, 0; blue, 0 }  ][line width=0.75]    (10.93,-3.29) .. controls (6.95,-1.4) and (3.31,-0.3) .. (0,0) .. controls (3.31,0.3) and (6.95,1.4) .. (10.93,3.29)   ;
\draw    (290,160) -- (289.88,186.07) ;
\draw [shift={(289.87,188.07)}, rotate = 270.27] [color={rgb, 255:red, 0; green, 0; blue, 0 }  ][line width=0.75]    (10.93,-3.29) .. controls (6.95,-1.4) and (3.31,-0.3) .. (0,0) .. controls (3.31,0.3) and (6.95,1.4) .. (10.93,3.29)   ;
\draw    (358.05,158) -- (398.23,179.07) ;
\draw [shift={(400,180)}, rotate = 207.67] [color={rgb, 255:red, 0; green, 0; blue, 0 }  ][line width=0.75]    (10.93,-3.29) .. controls (6.95,-1.4) and (3.31,-0.3) .. (0,0) .. controls (3.31,0.3) and (6.95,1.4) .. (10.93,3.29)   ;
\draw    (180,230) -- (179.88,256.07) ;
\draw [shift={(179.87,258.07)}, rotate = 270.27] [color={rgb, 255:red, 0; green, 0; blue, 0 }  ][line width=0.75]    (10.93,-3.29) .. controls (6.95,-1.4) and (3.31,-0.3) .. (0,0) .. controls (3.31,0.3) and (6.95,1.4) .. (10.93,3.29)   ;
\draw    (290.13,231.93) -- (290.01,258) ;
\draw [shift={(290,260)}, rotate = 270.27] [color={rgb, 255:red, 0; green, 0; blue, 0 }  ][line width=0.75]    (10.93,-3.29) .. controls (6.95,-1.4) and (3.31,-0.3) .. (0,0) .. controls (3.31,0.3) and (6.95,1.4) .. (10.93,3.29)   ;
\draw    (410,230) -- (409.88,256.07) ;
\draw [shift={(409.87,258.07)}, rotate = 270.27] [color={rgb, 255:red, 0; green, 0; blue, 0 }  ][line width=0.75]    (10.93,-3.29) .. controls (6.95,-1.4) and (3.31,-0.3) .. (0,0) .. controls (3.31,0.3) and (6.95,1.4) .. (10.93,3.29)   ;
\draw    (498.05,228) -- (538.23,249.07) ;
\draw [shift={(540,250)}, rotate = 207.67] [color={rgb, 255:red, 0; green, 0; blue, 0 }  ][line width=0.75]    (10.93,-3.29) .. controls (6.95,-1.4) and (3.31,-0.3) .. (0,0) .. controls (3.31,0.3) and (6.95,1.4) .. (10.93,3.29)   ;
\draw    (160,30) -- (160,340) ;
\draw    (50,20) -- (50,340) ;
\draw    (10,50) -- (656.87,49.2) ;

\draw (171,62.4) node [anchor=north west][inner sep=0.75pt]    {$\ket{2^{0} 3^{0}}$};
\draw (173,130.4) node [anchor=north west][inner sep=0.75pt]    {$d^{\dagger }\ket{2^{0} 3^{0}}$};
\draw (173,312.4) node [anchor=north west][inner sep=0.75pt]    {$\vdots $};
\draw (614,312.4) node [anchor=north west][inner sep=0.75pt]    {$\ddots $};
\draw (523,312.4) node [anchor=north west][inner sep=0.75pt]    {$\vdots $};
\draw (281,130.4) node [anchor=north west][inner sep=0.75pt]    {$c^{\dagger }\ket{2^{0} 3^{0}}$};
\draw (171,192.4) node [anchor=north west][inner sep=0.75pt]    {$\left( d^{\dagger }\right)^{2}\ket{2^{0} 3^{0}}$};
\draw (281,200.4) node [anchor=north west][inner sep=0.75pt]    {$d^{\dagger } c^{\dagger }\ket{2^{0} 3^{0}}$};
\draw (401,192.4) node [anchor=north west][inner sep=0.75pt]    {$\left( c^{\dagger }\right)^{2}\ket{2^{0} 3^{0}}$};
\draw (171,269.4) node [anchor=north west][inner sep=0.75pt]    {$\left( d^{\dagger }\right)^{3}\ket{2^{0} 3^{0}}$};
\draw (276,269.4) node [anchor=north west][inner sep=0.75pt]    {$\left( d^{\dagger }\right)^{2} c^{\dagger }\ket{2^{0} 3^{0}}$};
\draw (401,269.4) node [anchor=north west][inner sep=0.75pt]    {$d^{\dagger }\left( c^{\dagger }\right)^{2}\ket{2^{0} 3^{0}}$};
\draw (521,269.4) node [anchor=north west][inner sep=0.75pt]    {$\left( c^{\dagger }\right)^{3}\ket{2^{0} 3^{0}}$};
\draw (403,312.4) node [anchor=north west][inner sep=0.75pt]    {$\vdots $};
\draw (283,312.4) node [anchor=north west][inner sep=0.75pt]    {$\vdots $};
\draw (21,32.4) node [anchor=north west][inner sep=0.75pt]    {$k$};
\draw (65,31) node [anchor=north west][inner sep=0.75pt]   [align=left] {Eigenvalues};
\draw (170,32) node [anchor=north west][inner sep=0.75pt]   [align=left] {Eigenvectors};
\draw (21,62.4) node [anchor=north west][inner sep=0.75pt]    {$0$};
\draw (21,132.4) node [anchor=north west][inner sep=0.75pt]    {$1$};
\draw (21,202.4) node [anchor=north west][inner sep=0.75pt]    {$2$};
\draw (21,272.4) node [anchor=north west][inner sep=0.75pt]    {$3$};
\draw (101,62.4) node [anchor=north west][inner sep=0.75pt]    {$0$};
\draw (81,132.4) node [anchor=north west][inner sep=0.75pt]    {$-1,\ 1$};
\draw (71,202.4) node [anchor=north west][inner sep=0.75pt]    {$-2,\ 0,\ 2$};
\draw (57,272.4) node [anchor=north west][inner sep=0.75pt]    {$-3,\ -1,\ 1,\ 3$};
\draw (93,312.4) node [anchor=north west][inner sep=0.75pt]    {$\vdots $};
\draw (21,312.4) node [anchor=north west][inner sep=0.75pt]    {$\vdots $};

\end{tikzpicture}

    \caption{Eigenvalues and eigenvectors of $\mathcal{H}\lvert_{\mathbb{H}^{\odot k}}$ for $k=0,1,2,3,\dots$\space.}
    \label{dimer_eigen_table}
\end{figure}

\section{Normalization and the eigenbasis}
 Define $\hat{c}_k,d_k: \mathbb{H}^{\odot k} \rightarrow \mathbb{H}^{\odot k-1}$ as follows
\begin{equation}
    \hat{c}_k = \frac{1}{\sqrt{2k}}(\hat{a}_2 + \hat{a}_3), \quad \hat{d}_k = \frac{1}{\sqrt{2k}}(\hat{a}_2 - \hat{a}_3).
\end{equation}
Consider an arbitrary state
\begin{equation}
\ket{x} = \sum_{\alpha=0}^k x_{2^{k-\alpha}3^\alpha}\ket{2^{k-\alpha} 3^\alpha} \in \mathbb{H}^{\odot k}.
\end{equation}
A direct calculation yields:
\begin{eqnarray}
   \sqrt{2k}\, \hat{c}_k \ket{x} &=&  \sum_{\alpha=0}^{k-1}\Big(x_{2^{k-\alpha}3^\alpha}\sqrt{k-\alpha} + x_{2^{k-\alpha-1}3^{\alpha+1}}\sqrt{\alpha+1}\Big)\ket{2^{k-\alpha-1}3^\alpha} \\
   &&  \\
   \sqrt{2k}\, \hat{d}_k \ket{x} &=& \sum_{\alpha=0}^{k-1}\Big(x_{2^{k-\alpha}3^\alpha}\sqrt{k-\alpha} - x_{2^{k-\alpha-1}3^{\alpha+1}}\sqrt{\alpha+1}\Big)\ket{2^{k-\alpha-1}3^\alpha}
\end{eqnarray} 
Therefore 
\begin{align}
k\Big(\|\hat{c}_k\ket{x}\|^2+\|\hat{d}_k\ket{x}\|^2\Big) &= \sum_{\alpha=0}^{k-1}\Big(k-\alpha)|x_{2^{k-\alpha}3^\alpha}|^2 +(\alpha+1)|x_{2^{k-\alpha-1}3^{\alpha+1}}|^2\Big) \\
&= \sum_{\alpha=0}^k\Big(\alpha |x_{2^{k-\alpha}3^\alpha}|^2 + (k-\alpha)|x_{2^{k-\alpha}3^\alpha}|^2\Big) \\
&=k\sum_{\alpha=0}^k |x_{2^{k-\alpha}3^\alpha}|^2.
\end{align}
That is
\begin{align}
\|\hat{c}_k\ket{x}\|^2+\|\hat{d}_k\ket{x}\|^2 = \|\ket{x}\|^2,
\end{align}
giving a normalized sum.

The matrices for $\hat{c}_k,\hat{d}_k$ have $k$ rows and $k+1$ columns, and take the form
\begin{equation}
 \hat{c}_k = \frac{1}{\sqrt{2k}} \, \begin{pmatrix} 
    \sqrt{k} & \sqrt{1} \\
    & \sqrt{k-1} & \sqrt{2} \\
    & & \sqrt{k-2}& \ddots \\
    &&&\ddots &\sqrt{k-1}\\
    &&&&\sqrt{1}&\sqrt{k}
    \end{pmatrix}
\end{equation}
\begin{equation}
    \hat{d}_k = \frac{1}{\sqrt{2k}} \, \begin{pmatrix} 
    \sqrt{k} & - \sqrt{1} \\
    & \sqrt{k-1} & -\sqrt{2} \\
    & & \sqrt{k-2}& \ddots \\
    &&&\ddots &-\sqrt{k-1}\\
    &&&&\sqrt{1}&-\sqrt{k}
    \end{pmatrix}
\end{equation}
In particular, the $(2k)\times(k+1)$ block matrix created by stacking $\hat{c}_k$ on top of $d_k$ has columns that are mutually orthogonal. Therefore, this block matrix is a unitary immersion.

Computation of the coefficients of a vector $\ket{x}\in \mathbb{H}^{\odot k}$ can be performed recursively, as indicated in Fig.\ref{Fig2}.  Although this is a hierarchical computation, it does not constitute a speedup over the standard calculation of those coefficients (by taking the inner product of $\ket{x}$ with the columns of that unitary immersion). Indeed, it is easily seen that both schemes require $\mathcal{O}(k^3)$ arithmetical operations.

\begin{figure}[ht]
\centering
\begin{tikzpicture}[x=0.75pt,y=0.75pt,yscale=-1,xscale=2]

\draw    (20,40) -- (19.88,66.07) ;
\draw [shift={(19.87,68.07)}, rotate = 270.27] [color={rgb, 255:red, 0; green, 0; blue, 0 }  ][line width=0.75]    (10.93,-3.29) .. controls (6.95,-1.4) and (3.31,-0.3) .. (0,0) .. controls (3.31,0.3) and (6.95,1.4) .. (10.93,3.29)   ;
\draw    (20,110) -- (19.88,136.07) ;
\draw [shift={(19.87,138.07)}, rotate = 270.27] [color={rgb, 255:red, 0; green, 0; blue, 0 }  ][line width=0.75]    (10.93,-3.29) .. controls (6.95,-1.4) and (3.31,-0.3) .. (0,0) .. controls (3.31,0.3) and (6.95,1.4) .. (10.93,3.29)   ;
\draw    (50,40) -- (73.59,63.59) ;
\draw [shift={(75,65)}, rotate = 225] [color={rgb, 255:red, 0; green, 0; blue, 0 }  ][line width=0.75]    (10.93,-3.29) .. controls (6.95,-1.4) and (3.31,-0.3) .. (0,0) .. controls (3.31,0.3) and (6.95,1.4) .. (10.93,3.29)   ;
\draw    (80,110) -- (79.88,136.07) ;
\draw [shift={(79.87,138.07)}, rotate = 270.27] [color={rgb, 255:red, 0; green, 0; blue, 0 }  ][line width=0.75]    (10.93,-3.29) .. controls (6.95,-1.4) and (3.31,-0.3) .. (0,0) .. controls (3.31,0.3) and (6.95,1.4) .. (10.93,3.29)   ;
\draw    (120,110) -- (143.59,133.59) ;
\draw [shift={(145,135)}, rotate = 225] [color={rgb, 255:red, 0; green, 0; blue, 0 }  ][line width=0.75]    (10.93,-3.29) .. controls (6.95,-1.4) and (3.31,-0.3) .. (0,0) .. controls (3.31,0.3) and (6.95,1.4) .. (10.93,3.29)   ;

\draw (13,12.4) node [anchor=north west][inner sep=0.75pt]    {$\ket{x}$};
\draw (11,80.4) node [anchor=north west][inner sep=0.75pt]    {$\hat{d}_1\ket{x}$};
\draw (71,80.4) node [anchor=north west][inner sep=0.75pt]    {$\hat{c}_1\ket{x}$};
\draw (11,146.4) node [anchor=north west][inner sep=0.75pt]    {$\hat{d}_2\hat{d}_1\ket{x}$};
\draw (73,150.4) node [anchor=north west][inner sep=0.75pt]    {$ \hat{d}_2\hat{c}_1 \ket{x}$};
\draw (151,146.4) node [anchor=north west][inner sep=0.75pt]    {$\hat{c}_2\hat{c}_1\ket{x}$};
\draw (11,196.4) node [anchor=north west][inner sep=0.75pt]    {$\hat{d}_k \ldots\hat{d}_1\ket{x}$};
\draw (71,196.4) node [anchor=north west][inner sep=0.75pt]    {$\hat{d}_k \ldots\hat{d}_2 \hat{c}_1\ket{x}$};
\draw (201,196.4) node [anchor=north west][inner sep=0.75pt]    {$\hat{c}_k \ldots\hat{c}_1\ket{x}$};
\draw (161,192.4) node [anchor=north west][inner sep=0.75pt]    {$\dotsc $};
\draw (13,172.4) node [anchor=north west][inner sep=0.75pt]    {$\vdots $};
\draw (194,172.4) node [anchor=north west][inner sep=0.75pt]    {$\ddots $};
\draw (81,172.4) node [anchor=north west][inner sep=0.75pt]    {$\vdots $};

\end{tikzpicture}
\caption{A visualization of how a vector $\ket{x}$ can be represented in the orthonormal eigenbasis of $\mathcal{H}$, i.e., how the coefficients can be calculated by subsequent applications of $\hat{c}_k$ and $\hat{c}_k$. } 
\label{Fig2}
\end{figure}
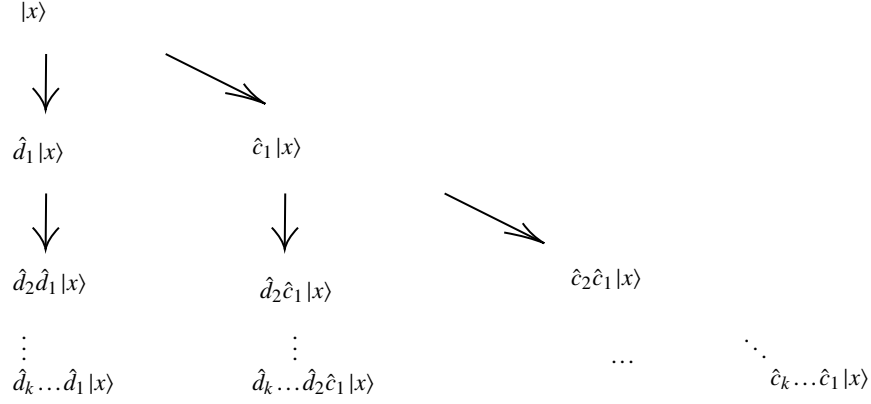

\section{The Fock space and coherent states}

Theorem \ref{TheorMain} enables us to discuss the properties of Hamiltonian (\ref{hopping hamiltonian}) as an operator in the Fock space (\ref{Fock}). We have observed that the dimer admits two sets of creation and annihilation operators: the first, stemming from $\hat{a}_2$ and  $\hat{a}_3$, and the second, stemming from  $\hat{c}$ and  $\hat{d}$. Indeed, both sets satisfy the BCCM as stated in (\ref{BCCM}) and in (\ref{def_c_and_d_clean}). This leads to two definitions of coherent states (CS). The first is given by 
\begin{equation}\label{CS0}
  \ket{w,z}_0 = e^{-\frac{|w|^2 +|z|^2}{2}} \exp(w\hat{a}_2^\dagger)
 \, \exp(z\hat{a}_3^\dagger)\, \ket{1},
\end{equation}
while the second is given by
\begin{equation}\label{CS1}
  \ket{w,z} = e^{-\frac{|w|^2 +|z|^2}{2}} \exp(w\hat{c}^\dagger)
 \, \exp(z\hat{d}^\dagger)\, \ket{1} =
  e^{-\frac{|w|^2 +|z|^2}{2}} \sum_{k =0}^{\infty} \sum_{m=0}^{k}\,\frac{w^mz^{k-m}}{m!(k-m)!} \, (\hat{c}^\dagger)^m
  (\hat{d}^\dagger)^{k-m} \, \ket{1}.
\end{equation}
However, these definitions are equivalent via a change of variables. Indeed, since all creation operators commute, we observe:
\begin{eqnarray*}
  \exp(w\hat{c}^\dagger)\, \exp(z\hat{d}^\dagger) &=& \exp\left(w\frac{\hat{a}_2^\dagger + \hat{a}_3^\dagger}{\sqrt{2}} + z \frac{\hat{a}_2^\dagger - \hat{a}_3^\dagger}{\sqrt{2}}\right) \\
   &=& \exp\left(\frac{w+z}{\sqrt{2}}\hat{a}_2^\dagger + \frac{w-z}{\sqrt{2}}\hat{a}_3^\dagger\right) = \exp\left(\frac{w+z}{\sqrt{2}}\hat{a}_2^\dagger\right) \exp\left(\frac{w-z}{\sqrt{2}}\hat{a}_3^\dagger\right).
\end{eqnarray*}
In summary, 
\begin{equation}\label{CSequiv}
  \ket{w,z} = \ket{\frac{w+z}{\sqrt{2}}, \frac{w-z}{\sqrt{2}}}_0.
\end{equation}
In what follows we will utilize definition (\ref{CS1}), remembering that all findings could be restated in terms of CS defined in (\ref{CS0}). As is well known, the CS are the eigenvectors of the annihilation operators, namely:
\begin{equation}\label{eig1}
\hat{c} \, \ket{w,z} = w \ket{w,z}, \quad \hat{d} \, \ket{w,z} = z \ket{w,z}. 
\end{equation}
Note that (\ref{theH_in_cd}) indicates
\begin{equation}\label{exp1}
\bra{w,z}\mathcal{H} \ket{w,z} = |w|^2 - |z|^2.
\end{equation}
Theorem \ref{TheorMain} indicates that  $\ket{w,z}$ is a superposition of the eigenvectors of $\mathcal{H}$. In fact, we have:
\begin{equation}\label{CS1grH}
 \mathcal{H}\, \ket{w,z} = e^{-\frac{|w|^2 +|z|^2}{2}} \sum_{k =0}^{\infty} \sum_{m=0}^{k}\,\frac{w^mz^{k-m}}{m!(k-m)!} \, (-k+2m) \,(\hat{c}^\dagger)^m
  (\hat{d}^\dagger)^{k-m} \, \ket{1}.
\end{equation}

Note that the set of eigenvalues of $\mathcal{H}$ is $\mathbb{Z}$, i.e. the set of all integers.  It is interesting to ask for the probability that the energy measurement will return the value $\alpha \in \mathbb{Z}$, i.e. $\mathcal{P}(E = \alpha)$, when the state is $\ket{w,z}$. A calculation shows that we have 
\begin{eqnarray}
   \mathcal{P}(E = \alpha\geq 0)  &=& e^{-|w|^2-|z|^2}\, |w|^{2\alpha} \sum_{n=0}^{\infty} \frac{|wz|^{2n}}{n!(n+\alpha)!}  \\
 \nonumber&& \\
   \mathcal{P}(E = \alpha<0)  &=& e^{-|w|^2-|z|^2}\, |z|^{2\alpha} \sum_{n=0}^{\infty} \frac{|wz|^{2n}}{n!(n-\alpha)!} 
\end{eqnarray}

\section{Two-site cat states induced by the dynamics furnished by $\mathcal{H}^2$}\label{Section_Cat}

The phenomenon of creation of cat states via nonharmonic evolution goes back to the seminal papers \cite{Milburn} and \cite{Yurke-Stoler}. Therein, a CS is evolved to a cat state via the dynamics induced by the square of the number operator. It is a straightforward deduction from those classical results that the evolution via $\hat{N}^2 = (\hat{a}_2^\dagger \hat{a}_2 + \hat{a}_3^\dagger \hat{a}_3)^2$ would periodically effect a cat state. Here, we demonstrate that cat states also result from the dynamics induced by $\mathcal{H}^2$, where $\mathcal{H}$ is given by (\ref{hopping hamiltonian}). The hypothesis is natural in light of the facts we have already established: first, that $\mathcal{H}$ is somewhat akin to the number operator, (\ref{theH_in_cd}) and, second, that CS may be expressed as a superposition of its eigenstates, (\ref{CS1}). However, this would not be obvious \emph{a priori} since in the standard form
\[
\mathcal{H}^2 = (\hat{a}_2^\dagger)^2\hat{a}_3^2 + (\hat{a}_3^\dagger)^2\hat{a}_2^2 - 2\hat{a}_2^\dagger\hat{a}_3^\dagger \hat{a}_2\hat{a}_3 + \hat{N}
\]
appears to entail more hopping processes than
\[
\hat{N}^2 =  - (\hat{a}_2^\dagger)^2\hat{a}_2^2 - (\hat{a}_3^\dagger)^2\hat{a}_3^2+ 2\hat{a}_2^\dagger\hat{a}_3^\dagger \hat{a}_2\hat{a}_3  +  \hat{N}. 
\] 
To establish creation of a cat state, we begin by observing that (\ref{CS1grH}) implies:
\begin{equation}\label{CS1grexpH}
 \ket{w,z,t} := e^{i\mathcal{H}^2 t}\, \ket{w,z} = e^{-\frac{|w|^2 +|z|^2}{2}} \sum_{k =0}^{\infty} \sum_{m=0}^{k}\,\frac{w^mz^{k-m}}{m!(k-m)!} \, e^{i(-k+2m)^2t} \,(\hat{c}^\dagger)^m
  (\hat{d}^\dagger)^{k-m} \, \ket{1}.
\end{equation}
Note that the evolution is periodic with period $2\pi$. Also, since $e^{i(-k+2m)^2\pi} = e^{ik\pi} =(-1)^k$, and taking into account the fact that $w^mz^{k-m} (-1)^k = (-w)^m(-z)^{k-m}$, we see that the evolution effects a change of sign in intervals that are an odd multiple of $\pi$. Thus,
\[
\ket{w,z, t+2\pi} =  \ket{w,z, t},\quad \ket{w,z,t+ \pi} =  \ket{-w,-z,t}.
\]
Next, we invoke the basic fact:
\begin{equation}\label{expH}
   e^{i(-k+2m)^2\pi/2} = e^{i k^2\pi/2}  =\frac{1}{\sqrt{2}}\left( e^{i\pi/4} +(-1)^k e^{-i\pi/4}\right).
\end{equation}
Substituting this into (\ref{CS1grexpH}),  we obtain:
\begin{equation}\label{alpha_doubling}
 \ket{w,z,t+ \pi/2} = \frac{e^{i\pi/4}}{\sqrt{2}} \, \ket{w,z,t} + \frac{e^{-i\pi/4}}{\sqrt{2}} \, \ket{-w,-z,t}. 
\end{equation}
Thus a CS periodically turns into a cat state. 
\vspace{.5cm}

\noindent \textsc{Remark.} An analysis of cat states on an infinite array of boson sites has been given in \cite{Fransson_Sanders_Sowa}. Cat state formation via tunneling has also been investigated in \cite{Wielinga_Milburn} in the context of Kerr-type media.

\section{Conclusion}
We have introduced an explicit method for the construction of the eigenstates of the boson dimer hopping Hamiltonian. In particular, it furnishes a new way to diagonalize the spin $x$-projection operator. The construction enabled us to gain insight into the dynamics of two-site coherent states. We have established that the dynamics induced by the square of the hopping Hamiltonian leads to formation of two-site cat states.

\begin{acknowledgments}
The authors thank Barry Sanders for helpful comments and for pointing out a reference relevant to this work.
\end{acknowledgments}

\newpage
\nocite{*}
\bibliographystyle{unsrturl}
\bibliography{references}

@article{dimer01,
  author = {J. Sicks and H. Rieger},
  title = {{The double-well Bose Hubbard model with nearest-neighbour and cavity-mediated long-range interactions}},
  note = {arXiv:2308.15915. \url{https://doi.org/10.1103/physreva.109.033317}}
}

@article{Zhang_Rieger,
  author = {C. Zhang and H. Rieger},
  title = {{Phase diagrams of the disordered Bose-Hubbard model with cavity-mediated long-range and nearest-neighbor interactions}},
  note = {\textsc{Eur. Phys. J. B} \textbf{93} (2020), 25. \url{https://doi.org/10.1140/epjb/e2019-100420-1}}
}

@article{dimer02,
  author = {E.M. Graefe and U. Guenther and H.J. Korsch and A.E. Niederle},
  title = {{A non-Hermitian PT-symmetric Bose-Hubbard model: eigenvalue rings from unfolding higher-order exceptional points}},
  note = {arXiv:0802.3164. \url{https://doi.org/10.1088/1751-8113/41/25/255206}}
}

@article{exact_diag_bh,
  author = {J.M. Zhang and R.X. Dong},
  title = {{Exact diagonalization: the Bose-Hubbard model as an example}},
  note = {\textsc{Eur. J. Phys.} \textbf{31} (2010), 591. \url{https://doi.org/10.1088/0143-0807/31/3/016}}
}

@article{number_theoretic_bose_hubbard,
  author = {A. Sowa and J. Fransson},
  title = {{Solving the Bose-Hubbard model in new ways}},
  note = {\textsc{Quantum} \textbf{6} (2022), 728. \url{https://doi.org/10.22331/q-2022-06-02-728}}
}

@article{nonlocal_coherent_states,
  author = {A. Sowa and J. Fransson},
  title = {{Nonlocal coherent states in an infinite array of boson sites}},
  note = {\textsc{Physica A: Statistical Mechanics and its Applications} \textbf{669} (2025), 130606. \url{https://doi.org/10.1016/j.physa.2025.130606}}
}

@article{Milburn,
  author = {G.J. Milburn},
  title = {{Quantum and classical Liouville dynamics of the anharmonic oscillator}},
  note = {\textsc{Phys. Rev. A} \textbf{33} (1986), 674. \url{https://doi.org/10.1103/physreva.33.674}}
}

@article{Yurke-Stoler,
  author = {B. Yurke and D. Stoler},
  title = {{Generating quantum mechanical superpositions of macroscopically distinguishable states via amplitude dispersion}},
  note = {\textsc{Phys. Rev. Lett.} \textbf{57} (1986), 13. \url{https://doi.org/10.1103/physrevlett.57.13}}
}

@article{Chen,
  author = {Chen Wang et al.},
  title = {{A Schrödinger cat living in two boxes}},
  note = {\textsc{Science} \textbf{352} (2016), 1087. \url{https://doi.org/10.1126/science.aaf2941}}
}

@article{Fransson_Sanders_Sowa,
  author = {J. Fransson and B. Sanders and A. Sowa},
  title = {{Macroscopically distinguishable superposition in infinitely many degrees of freedom}},
  note = {\textsc{International Journal of Theoretical Physics} \textbf{65:16} (2026). \url{https://doi.org/10.1007/s10773-025-06214-z}}
}

@article{Wielinga_Milburn,
  author = {B. Wielinga and G.J. Milburn},
  title = {{Quantum tunneling in a Kerr medium with parametric pumping}},
  note = {\textsc{Phys. Rev. A} \textbf{48} (1993), 2494. \url{https://doi.org/10.1103/physreva.48.2494}}
}

\end{document}